\documentclass[journal]{IEEEtran}
\usepackage{amsfonts}
\usepackage{slashbox}
\usepackage{amsmath}
\usepackage{multirow}
\usepackage{graphicx}
\usepackage{amsfonts}
\usepackage{amsmath,epsfig}
\usepackage{stfloats}
\usepackage{amssymb}
\usepackage{cite}
\usepackage{amsmath}
\usepackage{amssymb}
\usepackage{latexsym}
\usepackage{multirow}
\usepackage{epsfig}
\usepackage{graphics}
\usepackage{epstopdf}
\usepackage{mathrsfs}
\usepackage{algorithmic}
\usepackage{algorithm}
\usepackage{subfigure}
\usepackage{slashbox}
\usepackage{color}
\usepackage[all]{xy}
\usepackage{stfloats}
\usepackage{graphicx}

\newsavebox{\tablebox}

\newcommand{\beq}{\begin{equation}}
\newcommand{\enq}{\end{equation}}
\newcommand{\ben}{\begin{eqnarray}}
\newcommand{\enn}{\end{eqnarray}}
\newcommand{\bei}{\begin{itemize}}
\newcommand{\eni}{\end{itemize}}

\newcommand{\bm}[1]{\mbox{\boldmath{$#1$}}}

\newtheorem{theorem}{Theorem}
\newtheorem{lemma}{{Lemma}}
\newtheorem{proof}{Proof}
\newtheorem{remark}{{Remark}}
\allowdisplaybreaks[1]

\ifCLASSINFOpdf

\else

\fi

\begin{document}
\title{Joint Optimization of User Association, Subchannel Allocation, and Power Allocation  in Multi-cell Multi-association OFDMA Heterogeneous Networks}

\author{Feng~Wang, Wen~Chen,~\IEEEmembership{Senior~Member,~IEEE}, Hongying Tang, and \\
  Qingqing Wu,~\IEEEmembership{Student~Member,~IEEE}

\thanks{F.~Wang and W.~Chen are with Department of Electronic Engineering, Shanghai Jiao Tong University,
China, e-mail:\{fengwang217@gmail.com; wenchen@sjtu.edu.cn\}. H. Tang is with Science and Technology on Microsystem Laboratory,
Shanghai Institute of Microsystems and Information Technology, Chinese Academy of Sciences, Shanghai, China, e-mail:tanghy@mail.sim.ac.cn. Q. Wu is with  Department of Electronic
Engineering, Shanghai Jiao Tong University, China, and also with the School of ECE, Georgia Institute of Technology, Atlanta, GA, USA, email: wu.qq1010@gmail.com.

This work is supported in part by NSFC under Grant 61671294, in part by Guangxi NSF key project under Grant
2015GXNSFDA139037, and in part by Shanghai Key Fundamental Project under Grant 16JC1402900. Part of this work was presented in 2016 IEEE 83rd Vehicular Technology Conference.
}

}


\maketitle

\begin{abstract}
Heterogeneous network is a novel network architecture proposed in Long-Term-Evolution~(LTE), which highly increases the capacity and coverage compared with the conventional networks. However, in order to provide the
best services, appropriate resource management must be applied. In this paper, we consider the joint
optimization problem of user association, subchannel allocation, and power allocation for downlink transmission in Multi-cell Multi-association Orthogonal Frequency Division Multiple Access (OFDMA) heterogeneous networks. To solve the optimization problem, we first divide it into two subproblems: 1) user association and subchannel allocation for fixed power allocation; 2) power allocation for fixed user association and subchannel allocation. Subsequently, we obtain a locally optimal solution for the joint optimization problem by solving these two subproblems alternately. For the first subproblem, we derive the globally optimal solution based on graph theory. For the second subproblem, we obtain a Karush-Kuhn-Tucker (KKT) optimal solution by a low complexity algorithm based on the difference of two convex functions approximation (DCA) method. In addition,  the multi-antenna receiver case and the proportional fairness case are also discussed. Simulation results demonstrate that the proposed algorithms can significantly enhance the overall network throughput.
\end{abstract}
\IEEEpeerreviewmaketitle

\section{Introduction}
\IEEEPARstart{A}{s} a novel candidate technology in 5th generation (5G) wireless networks, heterogeneous network is proposed to increase network throughput and coverage, and reduce energy consumption. In the homogeneous networks, the transmission
power and coverage of each BS is similar. Nevertheless, macro and micro base stations~(BSs) with different transmission power and processing capability are deployed in heterogeneous networks to meet various communication demands \cite{86de1}. The micro BSs include picocell BSs, femtocell BSs \cite{86de4} and relays. Picocell BSs and femtocell BSs are connected to the network by wired backhaul and relays are connected to the network by wireless backhaul. A heterogeneous architecture brings in a rich topology, but the deployment of different low power BSs over existing macro BSs coverage causes severe interference, which poses new challenges on interference management and resource allocation.

Different kinds of user association schemes have been discussed for heterogeneous networks. To balance the traffic load between the BSs, range-expansion based scheme is proposed in \cite{five,86de7}, where a bias factor is used to balance the load in macro BSs and micro BSs. In \cite{2A}, WU \emph{et al}. propose a novel user association model with dual connectivity and constrained backhaul. In \cite{2B}, Siddique \emph{et al}. propose a channel-access-aware user association scheme  to enhance the spectral efficiency and achieve traffic load balancing. A Voronoi-based user association scheme is proposed in \cite{2C} to maximize the number of admitted users. Subchannel allocation is another hot issue in  Orthogonal Frequency Division Multiple Access~(OFDMA) heterogenous networks. In general cases, each subchannel will be allocated to the  user equipments (UEs) that has the best channel condition. Various methods are investigated for subchannel allocation, such as the worst user first (WUF) Greedy algorithm \cite{six} and the proportional fair method  \cite{seven}. A few recent literatures investigate joint optimization of user association  and subchannel allocation in multi-cell OFDMA networks, such as \cite{HetNets} and \cite{84}. In \cite{HetNets}, the user association and subchannel allocation problem is solved separately. In \cite{84}, the authors propose  an iterative method which only achieves a suboptimal solution. To our best knowledge, the optimal solution for the joint user association and subchannel allocation problem in multi-cell OFDMA networks  has not been achieved yet.

As the spectrum becomes rare and expensive, the co-channel deployment (CCD)  scheme, where all BSs operate on the full set of subchannels, are highly desirable \cite{84de6}. Some works focus on the resource allocation when the CCD is considered. In \cite{37}, a distributed power allocation method based on iterative water-filling (IW) is presented. In \cite{84de7}, Perez \emph{et al}. propose a dynamic algorithm to jointly  allocate subchannel and power to mitigate inter-cell interference. In \cite{1B}, Tabassum \emph{et al}. investigate the subchannel and power allocation in high signal-to-interference-noise-ratio~(SINR) regime. In \cite{1A}, Kim \emph{et al}. propose the joint subchannel allocation and power control based on polyblock outer approximation~(JSPPA) algorithm to get the optimal solution of the joint subchannel and power allocation problem. However, its computational complexity   increases exponentially  with the number of UEs and subchannels.
In addition, all the literatures above assume that one user can only be connected to one BS  in each time slot.  Recently,   Ghimire \emph{et al}.  assume that one user can be connected to multiple BSs, which is referred to as multi-association \cite{eight}. Intuitively, this will yield higher throughput since it allows for higher flexibility.  However, the authors adopt the exhaustive search for the user association and subchannel allocation problem in CCD networks.

In this paper, we also consider multi-association  scenario. In particular, we assume that one user can be connected to different BSs on different subchannels.  We maximize the weighted sum-rate for downlink transmission in multi-cell OFDMA heterogeneous networks. Our main contributions are summarized as follows:
\begin{itemize}
\item We develop a more general mathematical model considering both multi-association and CCD.  To the best of our knowledge, this is the first work that jointly  optimizes  user association, subchannel allocation, and power allocation for our considered system.

\item To solve this optimization problem, we first divide it into two subproblems. The first one is joint optimization of user association and subchannel allocation for a fixed power allocation. We transform it into an equivalent bipartite matching problem \cite{twentyone} and obtain the globally optimal solution by Hungarian algorithm \cite{twentytwo}. As far as we know, this is the first optimal solution for the joint problem in multi-cell networks.

\item The second subproblem is power allocation for  fixed  user association and subchannel allocation, which is transformed into a series of convex problems  based on the difference of two convex functions approximation (DCA)  method.  Then a low complexity algorithm is proposed to solve these problems. Moreover, we prove that the DCA method converges to a Karush-Kuhn-Tucker (KKT) optimal point under some mild conditions.

\item Based on that, we obtain a locally optimal solution for the joint optimization problem by an alternating optimization method. Simulation results show that our algorithm can achieve significant performance gain compared with the existing algorithms. We also find that the iteration number of the conventional Lagrange Dual method\cite{liuyong,kwan} is nearly 9 times larger than that of our low complexity method.

\item We also extend the problem into the multi-antenna receiver case and the proportional fairness case, respectively. Our algorithm  performs well when the UEs apply multi-antenna receivers, and it helps achieve a good tradeoff between throughput and fairness when considering proportional fairness.
\end{itemize}

The remainder of this paper is organized as follows. In Section II, we establish the multi-cell OFDMA heterogeneous networks model and formulate the weighted sum-rate maximization problem. In Section III, a joint optimization method is proposed to obtain the locally optimal solution by dividing the maximization problem into two subproblems, where the multi-antenna receiver case and the proportional fairness case are also discussed.   Simulation results are provided in Section IV.  Finally, the paper is concluded in Section V.

\section{System Description and Problem Formulation}
\subsection{System Description}

We consider an OFDMA multi-cell heterogeneous network in the downlink transmission as shown in Fig. 1. The total number of cells in the network is $N_c$. Each cell is associated with one macro BS in the center and $N_m-1$ uniformly deployed micro BSs. There are $N_u$ UEs  uniformly distributed in the network.  The total transmission time and frequency   band
are equally divided into multiple time slots and multiple subchannels, respectively. Each subchannel consists of several consecutive subcarriers. The channel is modeled to capture both the large-scale attenuation and the small-scale fading. Four basic assumptions are made in our system model as follows:

\begin{figure}[htp]
    \centering
    \includegraphics[width=3.5in]{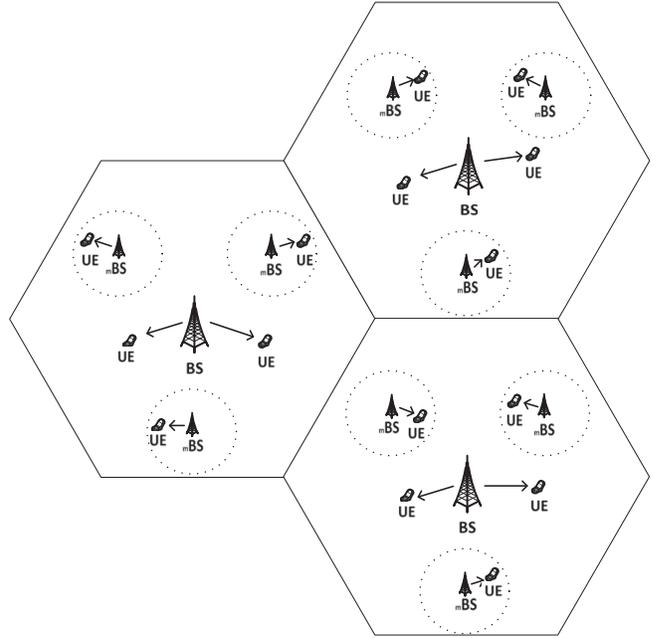}
    \caption{The system model of a downlink heterogeneous network.}\label{pareto}
\end{figure}

1) Back-haul connectivity: There is a centralized controller which is associated to all the BSs by optical fiber. Perfect channel state information (CSI) is available at the centralized controller.  The CSI can be collected by the following way: each BS broadcasts the pilot signal to all the UEs. Next, each UE  estimates the CSI and  sends it to the related BS via a feedback channel. Then all the BSs send the CSI to the centralized controller by optical fiber. Due to the high speed data exchange between the centralized controller and the BSs, the time cost of CSI overhead is negligible.

2) Co-channel deployment (CCD): All BSs (both macro BSs and micro BSs) operate on the full set of subchannels.  Although CCD may lead to severe interference, its system performance can outperform the spectrum splitting method in our system model,  as shown in the simulation results.

3) Multi-association: A UE can be served by multiple BSs in each time slot. Conventionally, a UE can only associate with one BS in each time slot \cite{101,86de12}. To further improve the network throughput, we assume that each UE can be served by different BSs on different subchannels.

4) Channel fading: The small-scale fading is assumed to be frequency selective and independent among different subchannels, while the channel in each subchannel is assumed to be flat fading. The channel coefficients remains unchanged within each time slot.

Denote $h_{i,j,k}^n$ as the channel coefficients from the \emph{j}th BS in the \emph{i}th cell to the \emph{k}th UE on the \emph{n}th subchannel,  where $j=1$  for macro BSs, and $j>1$ for micro BSs. Denote $N_r$ as the total number of subchannels. The total transmit power of the macro BSs and the micro BSs  are given by $P_b$ and $P_m$, respectively. Assuming that the \emph{n}th subchannel of the \emph{j}th BS in the \emph{i}th cell is allocated to the \emph{k}th UE, then the received signal-to-interference-noise-ratio~(SINR) on this subchannel can be expressed by

\begin{equation} \label{sinr1}
{\rm SINR}_{i,j,k}^n = \frac{{p_{i,j}^nh_{i,j,k}^n}}{{\sum\limits_{i',j' \ne i,j} {p_{i',j'}^nh_{i',j',k}^n  + {N_0}}  }},
\end{equation}
where $p_{i,j}^n$ is the power allocated on the \emph{n}th subchannel at the \emph{j}th BS in the \emph{i}th cell and $N_0$ is the variance of the Additive White Gaussian Noise~(AWGN). Then the data rate of the \emph{k}th UE received from the \emph{j}th BS in the \emph{i}th cell on the \emph{n}th subchannel in terms of bit/s/Hz is given by
\begin{equation}
R_{i,j,k}^n=\log(1+{\rm SINR}_{i,j,k}^n).
\end{equation}

Let  binary variable ${u_{i,j,k}}$   represent the  user association, where ${u_{i,j,k}} =1 $ if the \emph{k}th UE is associated with the \emph{j}th BS in the \emph{i}th cell  and ${u_{i,j,k}} =0 $ otherwise. Let binary variable $s_{i,j,k}^n$ represent the   subchannel allocation,  where it equals  $1$ if \emph{n}th subchannel of the \emph{j}th BS in the \emph{i}th cell is allocated to the \emph{k}th UE, and  equals $0$ otherwise. Then the sum-rate of the heterogeneous network is given as follows
\begin{equation}
{\sum\limits_{i = 1}^{{N_c}} {\sum\limits_{j = 1}^{{N_m}} {\sum\limits_{k = 1}^{{N_u}} {\sum\limits_{n = 1}^{{N_s}} {u_{i,j,k}s_{i,j,k}^nR_{i,j,k}^n} } } } }.
\end{equation}

\subsection{Problem Formulation}
In this paper, our goal is to jointly optimize  user association, subchannel allocation, and power allocation with the objective of maximizing the weighted sum-rate.  Mathematically, the considered problem is formulated as
\begin{subequations}  \label{eq:1}
\begin{align}
   {\rm{ }}\max_{u_{i,j,k},s_{i,j,k}^n,p_{i,j}^n} & \sum\limits_{i = 1}^{{N_c}} {\sum\limits_{j = 1}^{{N_m}} {\sum\limits_{k = 1}^{{N_u}} {\sum\limits_{n = 1}^{{N_s}} {{\omega _k}u_{i,j,k}s _{i,j,k}^nR_{i,j,k}^n} } } }            \label{eq:1A} \\
               \ \ {\rm s.t.}   \quad \ & {\rm{  }}\sum\limits_{i = 1}^{{N_c}} {\sum\limits_{j = 1}^{{N_m}} {s_{i,j,k}^n \le 1,{\rm{    }}\quad\forall k,n} } ,            \label{eq:1B} \\
     \qquad & \sum\limits_{k = 1}^{{N_u}} {s_{i,j,k}^n}  \le 1,{\rm{         }}\quad \forall i,j,n, \label{eq:1C} \\
     \qquad & s_{i,j,k}^n \in \{ 0,1\} ,{\rm{     }}\quad \forall i,j,k,n,     \label{eq:1D}     \\
     \qquad & u_{i,j,k} \in \{ 0,1\} ,{\rm{     }}\quad \forall i,j,k,       \label{eq:1E}    \\
     \qquad & \sum\limits_{n = 1}^{{N_s}} {p_{i,j}^n \le {P_{i,j}}} {\rm{,      }}\quad \forall i,j,  \label{eq:1F} \\
     \qquad & {\rm{0}} \le p_{i,j}^n \le P_{i,j}^n{\rm{,       }}\quad \forall i,j,n.  \label{eq:1G}
\end{align}
\end{subequations}
In problem~\eqref{eq:1},~\eqref{eq:1A} represents the weighted sum-rate of the heterogeneous network, where $\omega_k$ is the weight of the \emph{k}th UE. Constraint~\eqref{eq:1B} means that on each subchannel, one UE can be connected to at most one BS. Constraint~\eqref{eq:1C} means that on each subchannel, one BS can serve at most one UE. Transmit powers are constrained by both total power limits given in~\eqref{eq:1F} and spectral masks in~\eqref{eq:1G}. For the user association variable ${u_{i,j,k}}$, we do not impose the constraint $\sum\limits_{i = 1}^{{N_c}} {\sum\limits_{j = 1}^{{N_m}} {{u_{i,j,k}} \le 1,} } \ \forall k$, which means that multi-association is permissible in our formulation.
Note that our formulation in~\eqref{eq:1} does not include any instantaneous QoS constraints for individual UEs. When the channel condition of a certain UE is very week, supporting its instantaneous QoS will consume vast resources (bandwith and power) or even be infeasible. Furthermore, since we assume CCD in the network, supporting the QoS of the UEs with bad channels creates severe interference to other UEs.   Note that different weights are introduced to represent different priorities and provide different QoS. By giving each UE a weight, we not only consider the different QoS of the UEs but also avoid wasting of resources compared to the instantaneous QoS constraints.

Sovling problem~\eqref{eq:1} is challenging due to the existence of the binary variables  and the non-convex SINR structure. A direct method would involve an exhaustive search over all possible  user association and subchannel allocation, followed by finding the optimal power allocation for each of them. However, the complexity is exponential which makes the exhaustive search infeasible in practice. Moreover, even for  fixed user association and subchannel allocation, it is still difficult to optimize the power due to the non-convex structure. Therefore, the conventional convex and quasi-convex optimization methods are not applicable to obtaining the optimal solution of problem~\eqref{eq:1}.

\section{Joint Optimization of User Association, Subchannel Allocation and Power Allocation}
In this section, we propose an alternating optimization method to solve the joint optimization problem~\eqref{eq:1}. First, we divide the problem into two subproblems: 1) joint optimization of user association and subchannel allocation for a fixed power allocation, 2) power allocation for  fixed  user association and subchannel allocation. For the first subproblem,  the globally optimal solution is derived. For the second subproblem, we  will obtain a local optimal solution. Second, we obtain a locally optimal solution of the joint optimization problem by solving these subproblems  alternately.

\subsection{Joint Optimization of User Association and Subchannel Allocation for Fixed Power Allocation}
In this subsection, we show how to obtain the optimal  user association and subchannel allocation for given power allocation. In the conventional schemes, each UE can only associate with one BS in each time slot. Most previous works  only consider either user association \cite{five,86de7,2A,2B,2C} or subchannel allocation \cite{six,seven}.  A few recent work consider joint optimiztion of user association and subchannel allocation, such as \cite{HetNets} and \cite{84}. However, only a suboptimal solution is achieved in the two articles.  In the following, we will propose
a new user association and subchannel allocation method by exploiting the graph theory, which yields the optimal solution in multi-association system.

The joint user association and subchannel allocation problem is formulated as follows
\begin{subequations}  \label{eq:111}
\begin{align}
   {\rm{ }}\max_{u_{i,j,k},s_{i,j,k}^n} & \sum\limits_{i = 1}^{{N_c}} {\sum\limits_{j = 1}^{{N_m}} {\sum\limits_{k = 1}^{{N_u}} {\sum\limits_{n = 1}^{{N_s}} {{\omega _k}u_{i,j,k}s _{i,j,k}^nR_{i,j,k}^n} } } }            \label{eq:1A} \\
               \ \ {\rm s.t.}   \quad \ & {\rm{  }}\sum\limits_{i = 1}^{{N_c}} {\sum\limits_{j = 1}^{{N_m}} {s_{i,j,k}^n \le 1,{\rm{    }}\quad\forall k,n} } ,            \label{eq:1B} \\
     \qquad & \sum\limits_{k = 1}^{{N_u}} {s_{i,j,k}^n}  \le 1,{\rm{         }}\quad \forall i,j,n, \label{eq:1C} \\
     \qquad & s_{i,j,k}^n \in \{ 0,1\} ,{\rm{     }}\quad \forall i,j,k,n,     \label{eq:1D}     \\
     \qquad & u_{i,j,k} \in \{ 0,1\} ,{\rm{     }}\quad \forall i,j,k.       \label{eq:1E}
\end{align}
\end{subequations}
Denote $\rho _{i,j,k}^n = {u_{i,j,k}}s_{i,j,k}^n$,  problem~\eqref{eq:111} is simplified as
\begin{subequations}  \label{eq:2}
\begin{align}
    {\rm{ }}\max_{\rho_{i,j,k}^n}  & \sum\limits_{i = 1}^{{N_c}} {\sum\limits_{j = 1}^{{N_m}} {\sum\limits_{k = 1}^{{N_u}} {\sum\limits_{n = 1}^{{N_s}} {{\omega _k}{\rho}_{i,j,k}^nR_{i,j,k}^n} } } }            \label{eq:2A} \\
               \ \ {\rm s.t.} \ {\rm{  }} & \sum\limits_{i = 1}^{{N_c}} {\sum\limits_{j = 1}^{{N_m}} {{\rho}_{i,j,k}^n \le 1,{\rm{    }}\quad\forall k,n} } ,           \label{eq:2B} \\
     \qquad & \sum\limits_{k = 1}^{{N_u}} {{\rho}_{i,j,k}^n}  \le 1,{\rm{         }}\quad\forall i,j,n, \label{eq:2C}  \\
     \qquad & \rho _{i,j,k}^n \in \{ 0,1\} ,{\rm{     }}\quad \forall i,j,k,n.     \label{eq:2D}
\end{align}
\end{subequations}
One can observe that the summations and the constraints in the above problem are independent with respect to $n$, which suggests that problem~\eqref{eq:2} can be decomposed into $N_s$ subproblems. Without loss of generality, we concentrate on the $m$th subchannel for analysis, where  $m\in \{1, \cdots, N_s\}$, i.e.,
\begin{subequations}  \label{eq:3}
\begin{align}
 {\rm{ }}\max_{\rho_{i,j,k}^m} & \sum\limits_{i = 1}^{{N_c}} {\sum\limits_{j = 1}^{{N_m}} {\sum\limits_{k = 1}^{{N_u}} {{\omega _k}{\rho}_{i,j,k}^mR_{i,j,k}^m} } }             \label{eq:3A} \\
   \ {\rm s.t.} \ &  {\rm{  }}\sum\limits_{i = 1}^{{N_c}} {\sum\limits_{j = 1}^{{N_m}} {{\rho}_{i,j,k}^m \le 1,{\rm{    }}\quad\forall k,} } \label{eq:3B} \\
    \qquad & \sum\limits_{k = 1}^{{N_u}} {{\rho}_{i,j,k}^m}  \le 1,{\rm{         }}\quad\forall i,j. \label{eq:3C}
\end{align}
\end{subequations}
Note that, multi-association is not supported in problem~\eqref{eq:3}, since it only focuses on the $m$th subchannel. Namely, one UE can be connected to at most one BS and one BS can serve at most one UE, if there is only one subchannel. However,  problem~\eqref{eq:111} supports multi-association  due to that one user can be connected to different BSs on different subchannels. Problem~\eqref{eq:3} is a combinatorial optimization problem \cite{new} which can always be solved by   exhaustive search for all the possible cases.  Obviously, this leads
to a prohibitive computational complexity especially when $N_c$, $N_m$, and $N_k$ are large.  Therefore, it is necessary to explore a new way to solve problem~\eqref{eq:3}. In the following, we will  transform problem~\eqref{eq:3}  into an equivalent bipartite matching problem \cite{twentyone}.

We construct a bipartite graph \cite{twentyone} ${\rm{A = (}}{{\rm{V}}_{BS}} \times {V_{UE}}{\rm{,E)}}$, where the two sets of vertices, $V_{BS}$ and $V_{UE}$, are sets of BSs and UEs, respectively. Denote $E$ as the set of edges that connect to the vertices in the different set.  Vertice
 $v_{BS}(i,j)$ denotes the \emph{j}th BS in the \emph{i}th cell and vertice $v_{UE}(k)$ denotes the \emph{k}th UE. Let $e(i,j,k)$ denote the edge connecting $v_{BS}(i,j)$ and $v_{UE}(k)$, and $w(i,j,k)$ is the weight of $e(i,j,k)$. We use
$\left|  \cdot  \right|$ to represent the cardinality of a set. Then $\left| {{V_{BS}}} \right| = {N_c}{N_m}$, $\left| {{V_{UE}}} \right| = {N_{\rm{u}}}$ and $\left| {\rm{E}} \right| = {N_c}{N_m}{N_u}$. Given a graph G = (V, E), a matching M in G is a set of pairwise non-adjacent edges. That is, no two edges share a common vertex. Let $S_M$ be the set consisting of all possible matchings. According to the description above, if we denote $w(i,j,k) ={{\omega _k}R_{i,j,k}^m}$, then we can solve problem~\eqref{eq:3} by finding a set of edges ${E^*}({E^*} \in {S_M})$ in the bipartite graph, which maximizes the sum-weight of the edges in ${E^*}$. This can be explained as follows:
\begin{itemize}
\item The value of $\rho_{_{i,j,k}}^{m}$ can be equally represented by the selection of the edge in $E^*$. ${{e(i,j,k)}} \in {E^*}$ represents $\rho_{_{i,j,k}}^{m}=1$ and ${{e(i,j,k)}} \notin {E^*}$ represents $\rho_{_{i,j,k}}^{m}=0$.\ \

\item The constraints in problem~\eqref{eq:3} is equal to ${E^*} \in {S_M}$.\ \

\item The maximum weighted sum-rate in problem~\eqref{eq:3} is equal to the maximum sum-weight of the edges in ${E^*}$.
\end{itemize}

This bipartite matching problem is called maximum weighted bipartite matching~(MWBM) problem \cite{twentyone}. The Hungarian algorithm \cite{twentytwo} is a classical algorithm to solve the MWBM problem. By converting problem~\eqref{eq:3} to the MWBM problem, we can obtain the globally optimal solution in polynomial time. After obtaining the optimal $\rho _{i,j,k}^n$, one can derive the  corresponding optimal user association and subchannel allocation by the following equations:

\begin{subequations}
\begin{align}
{u_{i,j,k}} = & \left\{ {\begin{array}{*{20}{c}}
   {1,{\rm{   if  }}\sum\limits_{n = 1}^{{N_s}} {\rho _{i,j,k}^n}  > 0,}  \\
   {0,{\rm{    \qquad      otherwise}}{\rm{,}}}  \\
\end{array}} \right.      \\
s_{i,j,k}^n= & \rho _{i,j,k}^n.
\end{align}
\end{subequations}

Note that, if a UE is required to associate to one BS only, the existing works can only obtain a suboptimal solution for the joint user association and subchannel allocation problem. Therefore, the multi-association assumption not only helps us improve the network throughput, but also makes the problem more tractable.

\subsection{Power Allocation for Fixed  User Association and Subchannel Allocation}
In this subsection, we discuss power allocation under fixed  user association and subchannel allocation.
Denote ${\bf{p}}\triangleq {(p_{1,1}^1,,...,p_{1,1}^{{N_s}},p_{1,2}^1,...,p_{1,2}^{{N_s}},...,p_{{N_c},
{N_m}}^1,...,p_{{N_c},{N_m}}^{{N_s}})^T}\in \mathbb R^{{N_cN_MN_S}}$.  For given $\rho_{i,j,k}^n$, problem~\eqref{eq:1} is simplified as
\begin{subequations}  \label{eq:34}
\begin{align}
   {\rm{ }}\max_{p_{i,j}^n} &\sum\limits_{i = 1}^{{N_c}} {\sum\limits_{j = 1}^{{N_m}} {\sum\limits_{n = 1}^{{N_s}} {{w_{k_{i,j,n}^ * }}R_{i,j,k_{i,j,n}^ * }^n} } }({\bf{p}})            \label{eq:34A} \\
               \ {\rm s.t.}  \ & \sum\limits_{n = 1}^{{N_s}} {p_{i,j}^n \le {P_{i,j}}} {\rm{,         }}\quad\forall i,j,        \label{eq:34B} \\
      & {\rm{0}} \le p_{i,j}^n \le P_{i,j}^n{\rm{,       }}\quad\forall i,j,n, \label{eq:34C}
\end{align}
\end{subequations}
where $k_{i,j,n}^*$ represents the $k$th UE associated with the \emph{j}th BS in the \emph{i}th cell on the \emph{n}th subchannel, and it has been determined by the given  user association and subchannel allocation.

Our goal is to obtain the optimal $p_{i,j}^n$ that  maximizes the  weighted sum-rate while satisfying the power constraints in~\eqref{eq:34}.  Note that for the special case when $N_c=1$ and $N_m=1$, problem~\eqref{eq:34} can be optimally solved by  the conventional water-filling algorithm \cite{ten}. However, for the multi-cell scenario, the problem becomes much more complicated due to the existence of inter-cell and intra-cell interference. In this case,  any  power allocation change will bring  impact on the resulting interference  as well as the SINR. Therefore,  the conventional water-filling algorithm is not applicable any more.  To deal with this problem, the authors in \cite{37} propose the iterative water-filling (IW) algorithm:
With a fixed total power constraint in each BS and uniform power allocation initially, the first BS updates its power allocation by the classical water-filling method, treating signals transmitted from all the other BSs as noise. Then the same process will be done for all the BSs one after another, and so forth until the process converges. However, since each BS never considers its interference to  other BSs, the IW algorithm is doomed to be unable to  achieve an  ideal network throughput.   In the next, we will propose a novel method by exploiting its implicit DC structure in problem~\eqref{eq:34}.

Note that the objective function of problem~\eqref{eq:34} is differentiable and can be written as the difference of two concave functions~\eqref{xiugai1}and~\eqref{xiugai2}.
\begin{figure*}[hb]
\begin{equation}\label{xiugai1}
g({\bf{p}})\triangleq \sum\limits_{i = 1}^{{N_c}} {\sum\limits_{j = 1}^{{N_m}} {\sum\limits_{n = 1}^{{N_s}} {{w_{k_{i,j,n}^ * }}} } }{\log \left( {\sum\limits_{i' = 1}^{{N_c}} {\sum\limits_{j' = 1}^{{N_m}} {p_{i',j'}^nh_{i',j',k_{i,j,n}^*}^n} }  + {N_0}} \right)},
\end{equation}
\end{figure*}
\begin{figure*}[hb]
\begin{equation}\label{xiugai2}
h({\bf{p}})\triangleq \sum\limits_{i = 1}^{{N_c}} {\sum\limits_{j = 1}^{{N_m}} {\sum\limits_{n = 1}^{{N_s}} {{w_{k_{i,j,n}^ * }}} } }{\log \left( {\sum\limits_{i' = 1}^{{N_c}} {\sum\limits_{j' = 1}^{{N_m}} {p_{i',j'}^nh_{i',j',k_{i,j,n}^*}^n - p_{i,j}^nh_{i,j,k_{i,j,n}^*}^n} }  + {N_0}} \right)}.
\end{equation}
\hrulefill
\end{figure*}
Such a problem is recognized as the difference of two concave
functions (DC) programming problem, which can
be efficiently solved via the DCA method \cite{81,tang}. The main idea of the DCA method is replacing the minuend by its first order Taylor expansion around some point and then solving the resulting convex problem.  For problem~\eqref{eq:34}, it is approximated as the following  problem  at the $s$th iteration
\begin{subequations}  \label{eq:17}
\begin{align}
   {\rm{ }}\max_{\mathbf p}  \ \ & g({\bf{p}})-h({\bf{p}}[s-1])-\nabla {h^T}({\bf{p}}[s-1])({\bf{p}}-{\bf{p}}[s-1])             \label{eq:17A} \\
  \ {\rm s.t.} \ \ \ &  \sum\limits_{n = 1}^{{N_s}} {p_{i,j}^n \le {P_{i,j}}} {\rm{,         }}\quad\forall i,j,            \label{eq:17B} \\
& {\rm{0}} \le p_{i,j}^n \le P_{i,j}^n{\rm{,       }}\quad\forall i,j,n. \label{eq:17C}
\end{align}
\end{subequations}
Alg. 1 shows how the DCA method works for problem~\eqref{eq:34}.  Its  convergence  has been discussed in previous works. In \cite{81}, the authors prove the convergence of the DCA method. In \cite{90}, the authors have shown that it can converge to a point which satisfies KKT optimality conditions. However, this proof depends heavily on the specific structure of their mathematical problem.   In Lemma 1, we will prove that the DCA method converges to a KKT optimal point, as long as the non-convex problem satisfies some  mild conditions. For the best of our knowledge,  this is the first proof which shows that the DCA method  converges to a KKT optimal point for general cases.
\begin{algorithm}[] 
\caption{The DCA method for solving problem~\eqref{eq:34}} 
\begin{algorithmic}[1] 
\STATE Choose an initial feasible point ${\bf{p}}[0]$ and set $s=1$.\\
\STATE Solve problem~\eqref{eq:17} and obtain ${\bf{p}}[s]$ .
\STATE Increase $s$ and go to step 2 until ${\bf{p}}[s]$ converges.
\end{algorithmic}
\end{algorithm}

\begin{lemma}
For any maximization problem with a convex feasible set, if the objective function is differentiable and can be written as the difference of two concave functions,  then at least a KKT point can be obtained by  the DCA method.
\end{lemma}
\begin{proof}
Without loss of generality, consider the following non-convex problem:
\begin{equation}\label{con25}
\begin{split}
\max \ & {f_0}(\textbf{x})\\
 {\rm s.t.} \ \ & {\rm{  }}{f_i}(\textbf{x}) \le 0,\ \ i = 1,2,...m{\rm{,}}
\end{split}
\end{equation}
where $f_0$ can be written as the difference of two concave functions, i.e., $f_0(\textbf{x})=g(\textbf{x})-h(\textbf{x})$, and both $g(\textbf{x})$ and $h(\textbf{x})$ are concave functions. The feasible set created by the constraints is convex. Then we can introduce a new additional variable $t$ and express  problem~\eqref{con25} equivalently as
\begin{equation}
\begin{split}
\max \ & t\\
\ {\rm s.t.} \ \ & {\rm{  }}t-{f_0}(\textbf{x}) \le 0, \\
\  \quad \quad & {\rm{  }}{f_i}(\textbf{x}) \le 0,\ \ i = 1,2,...,m{\rm{.}}
\end{split}
\end{equation}
Denote ${{\bf{x}}^{(s)}}$ as the optimal solution of the \emph{s}th convex problem. Let $f(\textbf{x},t)\triangleq t-f_0(\textbf{x})=t-g(\textbf{x})+h(\textbf{x})$. According to the DCA method, we use $\tilde f(\textbf{x},t)$ approximate $f(\textbf{x},t)$ by replacing $h(\textbf{x})$ with its first order Taylor expansion around ${{\bf{x}}^{(s-1)}}$, i.e., $\tilde f(\textbf{x},t)=t-g(\textbf{x})+h({{\bf{x}}^{(s-1)}})+\nabla {h^T}({{\bf{x}}^{(s-1)}})(\textbf{x}-{{\bf{x}}^{(s-1)}})$, to form the \emph{s}th convex problem. In the next, we provide three properties of DCA.

1) The following inequality holds for arbitrary $\textbf{x}$ and ${{\bf{x}}^{(s-1)}}$ due to the concavity of $h(\textbf{x})$
\begin{equation}
h(\textbf{x})\le h({{\bf{x}}^{(s-1)}})+\nabla {h^T}({{\bf{x}}^{(s-1)}})(\textbf{x}-{{\bf{x}}^{(s-1)}}).
\end{equation}
Thus, we have the following result for arbitrary $\textbf{x}$
\begin{equation} \label{con1}
\begin{split}
 \ f(\textbf{x},t)&=t-g(\textbf{x})+h(\textbf{x}) \\
&   \le t-g(\textbf{x})+h({{\bf{x}}^{(s-1)}})+\nabla {h^T}({{\bf{x}}^{(s-1)}})(\textbf{x}-{{\bf{x}}^{(s-1)}})  \\
&   =\tilde f(\textbf{x},t).
\end{split}
\end{equation}

2) The optimal solution of the \emph{(s-1)}th convex problem satisfies the following equations
\begin{equation}\label{con2}
\begin{split}
&   \tilde f({{\bf{x}}^{(s-1)}},t)    \\
=&t-g({{\bf{x}}^{(s-1)}})+h({{\bf{x}}^{(s-1)}})  \\
& +\nabla {h^T}({{\bf{x}}^{(s-1)}})({{\bf{x}}^{(s-1)}}-{{\bf{x}}^{(s-1)}}) \\
=&   t-g({{\bf{x}}^{(s-1)}})+h({{\bf{x}}^{(s-1)}})  \\
=&   f({{\bf{x}}^{(s-1)}},t).
\end{split}
\end{equation}

3) The gradient of $\tilde f({{\bf{x}}},t)$ is
\begin{equation}
\begin{split}
& \ \quad  \nabla\tilde f({{\bf{x}}},t) \\
& =\nabla (t-g(\textbf{x})+h({{\bf{x}}^{(s-1)}})+\nabla {h^T}({{\bf{x}}^{(s-1)}})(\textbf{x}-{{\bf{x}}^{(s-1)}})) \\
&  =\nabla t-\nabla g(\textbf{x})+\nabla (\nabla {h^T}({{\bf{x}}^{(s-1)}})(\textbf{x}-{{\bf{x}}^{(s-1)}}))  \\
&  =\nabla t-\nabla g(\textbf{x})+\nabla h({{\bf{x}}^{(s-1)}}).
\end{split}
\end{equation}
Then we have the following equation
\begin{equation} \label{con3}
\begin{split}
 \nabla\tilde f({{\bf{x}}^{(s-1)}},t)&=\nabla t-\nabla g({{\bf{x}}^{(s-1)}})+\nabla h({{\bf{x}}^{(s-1)}}) \\
 &=\nabla f({{\bf{x}}^{(s-1)}},t).
\end{split}
\end{equation}
Note that, the equations~\eqref{con1},~\eqref{con2}, and~\eqref{con3} are the sufficient conditions for the convergence to a KKT point\cite{83}.  Therefore, the DCA method will converge to a KKT point.
\end{proof}

We can easily know that Alg.~1 will converge to a KKT point, since problem~\eqref{eq:34} actually satisfies the above constraints in Lemma~1.
Now the remaining  task  is to solve problem~\eqref{eq:17},  where Lagrange dual technique \cite{liuyong,kwan} can be employed. The Lagrangian function is
\begin{equation}\small
\begin{split}
\mathcal{L}({\bf{p}},\bm{\lambda})\triangleq & g({\bf{p}})-h({\bf{p}}[s-1])-\nabla {h^T}({\bf{p}}[s-1])({\bf{p}}-{\bf{p}}[s-1]) \\
&+ \sum\limits_{i = 1}^{{N_c}} {\sum\limits_{j = 1}^{{N_m}} {{\lambda _{i,j}}({P_{i,j}} - \sum\limits_{n = 1}^{{N_s}} {p_{i,j}^n} )} }.
\end{split}
\end{equation}
where $\bm{\lambda}$ is a Lagrange multiplier vector corresponding to the maximum transmit power constraint (10b). Then the dual optimization problem is given by
\begin{equation}\label{equ:solving}
\begin{split}
& \mathop {\min }\limits_{\bm{\lambda }} \mathop {\max }\limits_{\bm{p}} {\mathcal L}({\bm{p}},{\bm{\lambda }}) \\
& \ {\rm s.t.} \ {\lambda _{i,j}} \ge 0.
\end{split}
\end{equation}
The above dual problem can be solved iteratively by decomposing it into two nested loops: the inner loop that maximizes
{\bf{p}} for given \bm{\lambda}, and the outer loop that determines the optimal \bm{\lambda}.  In the following, we will discuss them in detail.

1) The inner loop:  Denote ${\bm{\lambda}^{(l)}}$  as the  $l$th iteration dual variable.  For given ${\bm{\lambda}^{(l)}}$, we can derive the following power allocation under the constraint (10c), by setting the derivative of $\mathcal{L}({\bf{p}},\bm{\lambda})$ with respect to $p_{i,j}^n$ to zero.
\begin{equation}\label{equ:split}
\begin{split}
p_{i,j}^n=\left[{\frac{1}{{{\lambda _{i,j}^{(l)}}+d_{i,j}^n}}-\frac{{\sum\limits_{(u,v)\ne (i,j)}{p_{u,v}^nh_{u,v,k_{i,j,n}^*}^n}+
{N_0}}}{{h_{i,j,k_{i,j,n}^*}^n}}}\right]_0^{P_{i,j}^n},
\end{split}
\end{equation}
where
\begin{equation} \label{con14}
\begin{array}{l}
\begin{split}
 d_{i,j}^n = & \sum\limits_{(i',j') \ne (i,j)}  {\left( {\frac{{h_{i,j,k_{i',j',n}^*}^n}}{{\sum\limits_{(u,v) \ne (i',j')} {p_{u,v}^n[s - 1]h_{u,v,k_{i',j',n}^*}^n}  + {N_0}}}} \right.}  \\
& \left. { - \frac{{h_{i,j,k_{i',j',n}^*}^n}}{{\sum\limits_{u,v} {p_{u,v}^nh_{u,v,k_{i',j',n}^*}^n}  + {N_0}}}} \right),
 \end{split}
 \end{array}
 \end{equation}
which is a taxation term related to the interference to other scheduled users. Since $p_{i,j}^n$ also appears on the right side in~\eqref{equ:split}, a closed-form power allocation expression cannot be obtained directly. However, a unique optimal power allocation can be obtained by the fixed point method as shown in \cite{3716de16}.

2) The outer loop: Once the optimal power allocation is achieved, the solution of the dual problem
 can be updated by the subgradient method as follows

\begin{equation}\label{con15}
{\lambda _{i,j}^{(l + 1)}} = {\left[ {{\lambda _{i,j}^{(l)}} + {\delta}\left( {\sum\limits_{n = 1}^N {p_{i,j}^n}  - {P_{i,j}}} \right)} \right]^ + },
\end{equation}
where ${\delta}$ is  a sufficiently
small step size. Since the dual problem is always a convex optimization problem, the subgradient method will converge to the globally optimal solution. The detailed algorithm to solve problem~\eqref{eq:17} is summarized as follows.

\begin{algorithm}[] 
\caption{The Lagrange Dual Algorithm for Solving  Problem~\eqref{eq:17}} 
\begin{algorithmic}[1] 
\STATE Set $l=1$. Initialize $\lambda _{i,j}^{(l)}$ and ${\delta}$.
\STATE For given $\lambda _{i,j}^{(l)}$, obtain the optimal $p_{i,j}^n$ from~\eqref{equ:split} by the fixed point method.
\STATE Update $\lambda _{i,j}^{(l+1)}$ by~\eqref{con15}.
\STATE Increase $l$ and go to step 2 until $\frac{{\left| {\lambda _{i,j}^{\left( {l + 1} \right)} - \lambda _{i,j}^{\left( l \right)}} \right|}}{{\lambda _{i,j}^{\left( {l + 1} \right)}}} \le \varepsilon $, for $\forall i,j.$
\end{algorithmic}
\end{algorithm}

The computational complexity of Alg.~2 is $O\left( {{K_\lambda }{K_P}{{({N_C}{N_m})}^2}{N_S}} \right)$, where ${K_\lambda }$ and ${K_P}$ are the required  number of iterations  for updating  $\bm{\lambda}$ and $\bf{p}$, respectively. From the simulation results in Section IV, we will see that  ${K_\lambda }$ is large, which inherently increases the computational complexity of Alg.~2.

In the remainder of this subsection, we will propose a low complexity algorithm to solve problem  (10) by avoiding the gradient descent search on the dual variable $\bm{\lambda}$. The basic idea behind this algorithm is that $\lambda_{i,j}^{(l)}$ and $p_{i,j}^{n,{(l)}}$  can be updated simultaneously in the $l$th iteration. More specifically,
in the \emph{l}th iteration, we first obtain the taxation term $d_{i,j}^{n,{(l)}}$ for all BSs according to~\eqref{con14} in parallel. Then we update $\lambda_{i,j}^{(l)}$ and $p_{i,j}^{n,{(l)}}$ according to the KKT conditions
\begin{subequations}\small  \label{eq:10}
\begin{align}
  p_{i,j}^{n,{(l)}}=& \left[{\frac{1}{{{\lambda_{i,j}^{(l)}}+d_{i,j}^{n,{(l)}}}}-\frac{{\sum\limits_{(u,v)\ne (i,j)}{p_{u,v}^{n,{(l-1)}}h_{u,v,k_{i,j,n}^*}^n}+{N_0}}}{{h_{i,j,k_{i,j,n}^*}^n}}}\right]_0^{P_{i,j}^n},  \label{eq:10A}\\
& {\lambda_{i,j}^{(l)}}({P_{i,j}} - \sum\limits_{n = 1}^{{N_s}} {p_{i,j}^{n,{(l)}}} ) = 0{\rm{ }}, \quad \forall i,j,     \label{eq:10B} \\
   &   \sum\limits_{n = 1}^{{N_s}} {p_{i,j}^{n,{(l)}} \le {P_{i,j}}} {\rm{,      }}\quad \forall i,j,  \label{eq:10C} \\
   &   {\lambda _{i,j}^{(l)}}\ge 0, \quad \forall i,j.
\end{align}
\end{subequations}

In theorem 1, we will prove that there always exists a unique $\lambda_{i,j}^{(l)}$ which satisfies the conditions in~\eqref{eq:10}. Furthermore, since $\sum\limits_{n=1}^{{N_s}} {p_{i,j}^{n,{(l)}}}$ is  decreasing with respect to $\lambda_{i,j}^{(l)}$, we can efficiently obtain this $\lambda_{i,j}^{(l)}$ via the bisection method \cite{bisec}. Our low complexity algorithm to solve problem~\eqref{eq:17} is summarized as Alg. 3. Its computational complexity is $O\left( {{K_T}{{({N_C}{N_m})}^2}{N_S}} \right)$, where $K_T$ is the number of iterations.

\begin{theorem}
Given $p_{i,j}^{n,{(l-1)}}$, there exists a unique $\lambda_{i,j}^{(l)}$ which satisfies the conditions in~\eqref{eq:10}.
\end{theorem}
\begin{proof}
From~\eqref{eq:10A}, it is easy to  know that $\sum\limits_{n=1}^{{N_s}} {p_{i,j}^{n,{(l)}}}$ is strictly decreasing with respect to $\lambda_{i,j}^{(l)}$, when $\sum\limits_{n=1}^{{N_s}} {p_{i,j}^{n,{(l)}}}>0$. For given $p_{i,j}^{n,{(l-1)}}$, set ${\tilde \lambda _{i,j}}^{(l)}=0$ and calculate $\tilde p_{i,j}^n$ by~\eqref{eq:10A}. Then we have the following two cases.

1) For the case $\sum\limits_{n = 1}^{{N_s}} {\tilde p_{i,j}^n}  \le {P_{i,j}}$, $\tilde {\lambda}_{i,j}^{(l)}=0$ satisfies~\eqref{eq:10}. Furthermore, for any $\lambda_{i,j}^{(l)}>0$, we have $\sum\limits_{n = 1}^{{N_s}} {p_{i,j}^{n,{(l)}}} < {P_{i,j}}$ since  $\sum\limits_{n=1}^{{N_s}} {p_{i,j}^{n,{(l)}}}$ is strictly decreasing with respect to $\lambda_{i,j}^{(l)}$. Therefore, we have $\lambda _{_{i,j}}^{(l)}({P_{i,j}} - \sum\limits_{n = 1}^{{N_s}} {p_{i,j}^{n,{(l)}}} ) > 0$, which violates  equation~\eqref{eq:10B}. Hence we know that $\tilde {\lambda}_{i,j}^{(l)}=0$ is the unique $\lambda_{i,j}^{(l)}$ which satisfies the conditions in~\eqref{eq:10}.

2) For the case $\sum\limits_{n = 1}^{{N_s}} {\tilde p_{i,j}^n}>{P_{i,j}}$, due to the monotonicity of $\sum\limits_{n=1}^{{N_s}} {p_{i,j}^{n,{(l)}}}$, there exists a unique $\lambda_{i,j}^{(l)}>0$ which makes $\sum\limits_{n = 1}^{{N_s}} {p_{i,j}^{n,{(l)}}}={P_{i,j}}$ hold.  Furthermore, it is clear that this $\lambda_{i,j}^{(l)}$ satisfies the other conditions in~\eqref{eq:10}.
Theorem 1 is proved.
\end{proof}

\begin{algorithm}[] 
\caption{The Low Complexity Algorithm for Solving Problem~\eqref{eq:17}} 
\begin{algorithmic}[1] 
\STATE Choose an initial feasible point $p_{i,j}^{n,{(0)}}$ and set $l=1$.
\STATE Update $d_{i,j}^{n,{(l)}}$ according to~\eqref{con14} for all BSs.
\STATE Update $\lambda_{i,j}^{(l)}$ via the bisection method, then update $p_{i,j}^{n,{(l)}}$ according to~\eqref{eq:10A}.
\STATE Increase $l$ and go to step 2 until $\frac{{\left| {\lambda _{i,j}^{\left( {l+1} \right)} - \lambda _{i,j}^{\left( {l} \right)}} \right|}}{{\lambda _{i,j}^{\left( {l+1} \right)}}} \le \varepsilon,$ for $\forall i,j.$
\end{algorithmic}
\end{algorithm}
Note that once Alg. 3 converges, one will obtain a KKT point for problem~\eqref{eq:17}. Moreover, it is also a  globally optimal solution, according to the convexity of problem~\eqref{eq:17}. Deriving conditions under which Alg. 3 converges is intractable, although convergence has always been observed in our simulation results in Section IV.

\subsection{Joint Optimization of User Association, Subchannel Allocation, and Power Allotcation}
In the previous subsections, we have obtained the globally optimal  user association and subchannel allocation for fixed power allocation, and a KKT optimal power allocation for fixed user association and subchannel allocation. Now, we propose an alternating optimization method for the joint design of user association, subchannel allocation, and power allotcation, involving iterations between the Hungarian algorithm and Alg. 3 until convergence. The overall algorithm is summarized as follows,  whose convergence is guaranteed by Theorem $2$.
\begin{algorithm}[] 
\caption{The Joint Optimization Algorithm for Solving problem~\eqref{eq:1}} 
\begin{algorithmic}[1] 
\STATE Initialize ${\bf{p}}[0]$ uniformly and set $i=1$.
\STATE Update ${\bm{\rho}}[i]$ by the Hungarian algorithm for ${\bf{p}}[i-1]$.
\STATE Form the $i$-th approximated convex problem around ${\bf{p}}[i-1]$,  and solve this convex problem by Alg. 3 to update ${\bf{p}}[i]$.\\
\STATE Increase $i$ and go to step 2 until the weighted sum-rate converges.
\end{algorithmic}
\end{algorithm}

\begin{theorem}
The convergence of Alg. 4 can always be guaranteed.
\end{theorem}
\begin{proof}
From the above subsections, one can see the weighted sum-rate can only change in step 2 and step 3. In step 2, the Hungarian algorithm finds the globally optimal  user association and subchannel allocation. In step 3, Alg. 3 finds the globally optimal solution for the $i$-th approximated convex problem. Therefore, the sequence of iterations produces a monotonically increasing weighted sum-rate. Meanwhile, it is obvious that the weighted sum-rate has an upper bound for finite power constraints. Thus, the convergence of Alg. 4 is guaranteed.
\end{proof}

\begin{remark}
Note that even if we can obtain the globally optimal power for given  user association and subchannel allocation, the alternating optimization method  still  converges to a locally optimal solution. The globally optimal solution of problem~\eqref{eq:1} needs exhaustive search for all the possible  user association and subchannel allocation and get the optimal power allocation
for each possible case, which is impossible in practice.
\end{remark}

In the following, we will analyze the complexity of Alg. 4. In each iteration, the Hungarian algorithm is adopted in step 2 with complexity given by $O\left( {{N_u}^3{N_s}} \right)$, and Alg. 3 is adopted in step 3 with complexity given by $O\left( {{K_T}{{({N_C}{N_m})}^2}{N_S}} \right)$ as shown in the above subsection. Denote $K_J$ as the average iteration number, then the computational complexity of Alg. 4 is given by $O\left( {{K_J}{N_u}^3{N_s} + {K_J}{K_T}{{\left( {{N_c}{N_m}} \right)}^2}{N_s}} \right)$. From the simulation results,  we can observe that both $K_J$ and $K_T$ are very small. Therefore, Alg. 4 has a relatively low computational complexity.

\subsection{Multi-antenna Receivers}
In this subsection, we consider the case that all UE have the same number of antennas, denoted as $N_a$ ($N_a>1$). Note that by adopting multiple antennas, the receiver will obtain different SINR by different strategies. We focus on two different strategies: the maximum ratio combining (MRC) strategy and the interference rejection combining (IRC) strategy. According to what strategies the UEs apply, we partition them into two categories: MRC  receivers and IRC receivers.  In practice,  MRC receivers and IRC receivers  coexist  in the network, which are indistinguishable to the BSs.

Denote $\bm{h}_{i,j,k}^n\in \mathcal C^{N_a}$ as the channel coefficients vector from the \emph{j}th BS in the \emph{i}th cell to the \emph{k}th UE on the \emph{n}th subchannel. Assuming that the \emph{k}th UE is associated with the \emph{j}th BS in the \emph{i}th cell on the \emph{n}th subchannel, then the received signal in \emph{k}th UE is given by:
\begin{equation}\label{equ:ue}
 {\bf{y}}_{i,j,k}^n = {\bf{h}}_{i,j,k}^n\sqrt {p_{i,j}^n}x_{i,j}^n + \sum\limits_{i',j' \ne i,j} {{\bf{h}}_{i',j',k}^n\sqrt {p_{i',j'}^n}x_{i',j'}^n}  + {{\bf{n}}_{\bf{0}}},
\end{equation}
where $x_{i,j}^n$ is the data from the \emph{j}th BS in the \emph{i}th cell on the \emph{n}th subchannel,  and ${\bf{n}}_{\bf{0}}\sim \mathcal C\mathcal N(\mathbf 0, N_0\mathbf I_{N_a})$ is the complex AWGN vector.

For MRC receivers, the received signals in~\eqref{equ:ue} is multiplied by the channel coefficient vector, i.e.,
\begin{equation}
\begin{split}
&{({\bf{h}}_{i,j,k}^n)^H}{\bf{y}}_{i,j,k}^n={({\bf{h}}_{i,j,k}^n)^H}{\bf{h}}_{i,j,k}^n\sqrt{p_{i,j}^n} x_{i,j}^n  \\
& \qquad +\sum\limits_{i',j'\ne i,j}  {{{({\bf{h}}_{i,j,k}^n)}^H}{\bf{h}}_{i',j',k}^n\sqrt{p_{i',j'}^n}x_{i',j'}^n}+{({\bf{h}}_{i,j,k}^n)^H}{{\bf{n}}_{\bf{0}}},
\end{split}
\end{equation}
then the corresponding  SINR  is given by
\begin{equation}\label{equ:sinr} \small
{\rm SINR}_{i,j,k}^n = \frac{{p_{i,j}^n{{\left| {{{({\bf{h}}_{i,j,k}^n)}^H}{\bf{h}}_{i,j,k}^n} \right|}^2}}}{{\sum\limits_{i',j' \ne i,j} {p_{i',j'}^n{{\left| {{{({\bf{h}}_{i,j,k}^n)}^H}{\bf{h}}_{i',j',k}^n} \right|}^2} + {{\left| {{{({\bf{h}}_{i,j,k}^n)}^H}{{\bf{n}}_{\bf{0}}}} \right|}^2}} }}.
\end{equation}
Comparing formula~\eqref{equ:sinr} with formula~\eqref{sinr1}, we know that the SINR of the multi-antenna system has the same structure with that of the single-antenna system. Therefore, Alg. 4 can also be directly applied to the multi-antenna system  for the MRC receivers.

For IRC receivers, the received signals in~\eqref{equ:ue} is multiplied by some predefined vector, i.e.,
\begin{equation}
\begin{split}
{{\bf{w}}^H}{\bf{y}}_{i,j,k}^n=&{{\bf{w}}^H}{\bf{h}}_{i,j,k}^n\sqrt{p_{i,j}^n} x_{i,j}^n  \\
&+\sum\limits_{i',j'\ne i,j} {{{{\bf{w}}}^H}{\bf{h}}_{i',j',k}^n\sqrt{p_{i',j'}^n}x_{i',j'}^n}+{{\bf{w}}^H}{{\bf{n}}_{\bf{0}}},
\end{split}
\end{equation}
where {\small${{\bf{w}}^H} = {({\bf{h}}_{i,j,k}^n)^H}{(\sum\limits_{i',j' \ne i,j} {p_{i',j'}^n{\bf{h}}_{i',j',k}^n{{({\bf{h}}_{i,j,k}^n)}^H} + {N_0}{\bf{I}}} )^{ - 1}}$}. Thus, the received SINR  is given by
\begin{equation}\label{equ:is}
{\rm SINR}_{i,j,k}^n = \frac{{p_{i,j}^n{{\left| {{{{\bf{w}}}^H}{\bf{h}}_{i,j,k}^n} \right|}^2}}}{{\sum\limits_{i',j' \ne i,j} {p_{i',j'}^n{{\left| {{{{\bf{w}}}^H}{\bf{h}}_{i',j',k}^n} \right|}^2} + {{\left| {{{{\bf{w}}}^H}{{\bf{n}}_{\bf{0}}}} \right|}^2}} }}.
\end{equation}
From~\eqref{equ:is}, we can see that neither the numerator nor the denominator is convex. Therefore, the network throughput in this case does not have a DC structure.  However, since the BSs do not know which kind of receiver that the UEs apply in advance, they will continue to allocate power as if  the UEs applied MRC receivers. In the following simulation results, an interesting observation is that Alg. 4 behaves even better when applying to those IRC receivers.

\section{Simulation results}
In this section, we provide simulation results to demonstrate the effectiveness of the proposed methods in Multi-cell OFDMA Heterogeneous Networks. We consider a heterogeneous network consisting of $7$ cells, where one macro BS is deployed at the center and three micro BSs are randomly distributed in each cell. The power of each macro BS and micro BS are 46dBm and  30dBm, respectively. The bandwidth of each subcarrier is 15kHZ and each subchannel consists of 12 subcarriers. The whole frequency band is divided into $50$ subchannels. Unless otherwise specified, the inter-site distance (ISD) between macro BSs is set to be $500$ meters so that the network is interference limited. For the large-scale fading, the distance-dependent path loss in dB is modeled as $PL_{NLOS}=128.1+37.6\log_{10}(d)$, where $d$ is the distance from the user to the BS in kilometers. The log-normal shadowing is considered with $\sigma_{shadow}=10$dB and the penetration loss is assumed to be $20$dB. The small-scale fading is modeled as the normalized Rayleigh fading. The noise power spectral density is set to be $-174$dBm/Hz.

\subsection{Joint User Association and Subchannel Allocation}
In this subsection, we compare the throughput between different joint user association and subchannel allocation schemes. The number of UEs in each cell changes from 10 to 50. In Fig. 2, we compare the network throughput of the following three methods: 1) the Hungarian algorithm (HA): our algorithm which can obtain the globally optimal solution; 2) JO1: the joint optimization method proposed in \cite{84}; 3) JO2: the joint optimization method proposed in \cite{HetNets}. Uniform power allocation is performed in all the BSs. From the comparison, one can observe that our method outperforms the existing two joint optimization methods in multi-cell networks.
\begin{figure}[htp]
    \centering
    \includegraphics[width=3.5in]{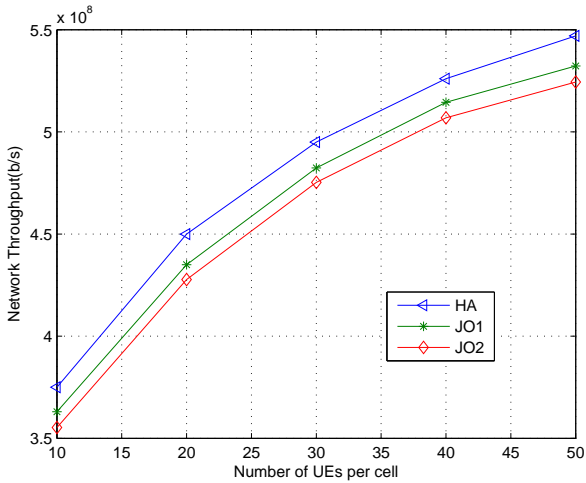}
    \caption{The network throughput versus the number of UEs per cell of HA, JO1 and JO2.}\label{paretonine}
\end{figure}

\subsection{Joint  User Association, Subchannel Allocation, and Power Optimization}
In this subsection, we  show the performance of the network throughput versus the number of UEs in each cell.
The number of UEs in each cell changes from 10 to 50. In Fig. 3, we compare the network throughput of the following five methods: 1) Alg. 4;  2) statistical CSI (SCSI): resource allocation by Alg.4 when the the intercell CSI is statistical, and the mean value is used instead of the instantaneous intercell CSI; 3) BPA: the belief propagation algorithm proposed in \cite{84} ; 4) IW: power allocation by the IW algorithm for fixed user association based on cell range expansion; 5) Static full spectral reuse (SFSR): uniform power allocation for fixed user association by cell range expansion. From the comparison, we can observe that Alg.~4 performs best among  all five schemes.  It should be mentioned that, when the intercell CSI is statistical,  Alg. 4 still outperforms the other three schemes. This result can be explained by the fact that  the interference from other cells is weaker compared with the intra-cell interference.  Hence, the SCSI method can be regarded as an alternative option in order to decrease the CSI overhead. Compared with BPA, we  can observe the throughput gain by Alg.4 and  SCSI due to the assumption of multiple-association and continuous power allocation.

In the following, we compare the computational complexity of the five methods. As mentioned above, the complexity of Alg. 4 and  SCSI  are both $O\left( {{K_J}{N_u}^3{N_s} + {K_J}{K_T}{{\left( {{N_c}{N_m}} \right)}^2}{N_s}} \right)$. The complexity of the BPA method is $O\left( {KW{H_v}^2 + KW{H_f}^2} \right)$, where $K$ is the average iteration number, $W$ is the avearage number of the scheduling options in variable nodes, $H_v$ is the average number of neighboring factor nodes of a variable node and $H_f$ is the average number of a factor node's neighboring variable nodes. The complexity of the IW algorithm and the SFSR method is $O\left( {{N_c}{N_m}{N_u} + K{{\left( {{N_c}{N_m}} \right)}^2}{N_s}} \right)$ and $O\left( {{N_c}{N_m}\left( {{N_u} + {N_s}} \right)} \right)$, respectively. In conclusion, the complexity of Alg. 4 and SCSI is higher than that of the IW algorithm and the SFSR method. Note that, we can't compare the complexity between our algorithm with the BPA method due to the fact that  $W$, $H_v$ and $H_f$ are not estimated in \cite{84}.
\begin{figure}[htp]
    \centering
    \includegraphics[width=3.5in]{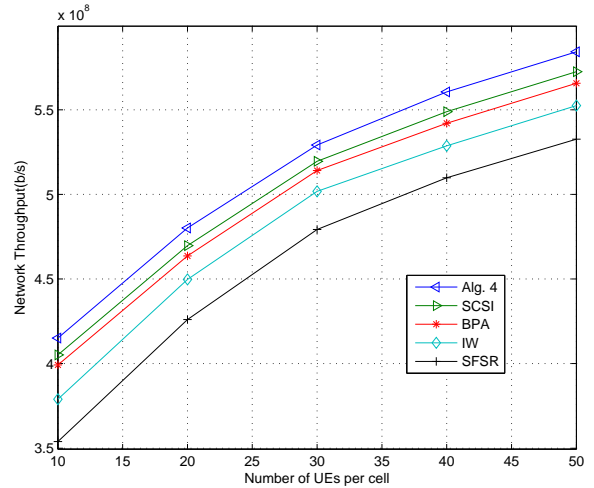}
    \caption{The network throughput versus the number of UEs per cell of Alg. 4, SCSI, BPA, IW, and SFSR.}
\end{figure}

In additon, if multi-association is not considered, the globally optimal solution of the joint optimization problem  can be obtained by the following two steps: 1) exhaustive search for all the possible user association; 2) get the optimal subchannel and power allocation for each possible case by the JSPPA algorithm \cite{1A}. The complexity in exhaustive search is $O\left( {{{\left( {{N_c}{N_m}} \right)}^{{N_u}}}} \right)$. The complexity of the JSPPA algorithm is $O\left( {{R_1}{{({N_e})}^{{N_c}{N_s}}}} \right)$, where $R_1$ is the number of iterations and $N_e$ is number of UEs in each cell. The globally optimal algorithm would take an unrealistically long time to return the globally optimal solution for  practical multi-cell networks. Therefore, we  compare Alg. 4 with the globally optimal algorithm in a single-cell scenario, where six  UEs are distributed uniform in the cell. The total number of subchannels is four. The transmit power of the micro BS is 30dBm and the transmit power of the macro BS changes from 40dBm to 46dBm. Fig. 4 shows the performance of the network throughput versus the transmit power of the macro BS.  From the simulation, we can see that the throughput gap between Alg. 4 and the above globally optimal (GO) algorithm is less than 3\%, which is negligible considering the huge computational complexity  reduction.

\begin{figure}[htp]
    \centering
    \includegraphics[width=3.5in]{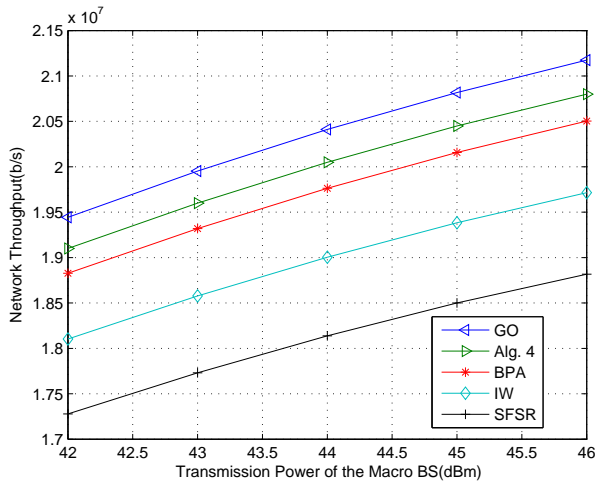}
    \caption{The network throughput of GO, Alg. 4, BPA, IW, and SFSR.}
\end{figure}

In the following, we compare Alg. 4 with the interference avoidance strategy where the BSs in each cell orthogonally utilize the resource. In Fig. 5, we compare the following three methods: 1) Alg. 4; 2) SFSR; 3) spectrum splitting (SS): subchannel allocation based on cell range expansion and power allocation based on classical water-filling algorithm \cite{ten} in each cell. Two scenarios are considered: the urban scenario where the ISD between macro BSs is 500m and the rural scenario where the ISD between macro BSs is 2000m. The number of UEs in each cell is $30$. The x-axis represents the ratio of subchannels used by micro BSs among all 50 subchannels for SS method, which are divided equally to $3$ micro BSs. As expected, the network throughputs of all the methods decrease with the increasing of ISD due to the increasing path loss. We can see that Alg. 4 is always better than the spectrum splitting policy. Another observation is that the best performance of SS is better than SFSR in the urban scenario. However, SFSR always outperforms SS in the rural scenario. The best performance of SS can be achieved at the ratio $18/50$, which means that in each cell, the macro BS occupy $32$ subchannels and every micro BS occupy $6$ subchannels.
\begin{figure}[htp]
    \centering
    \includegraphics[width=3.5in]{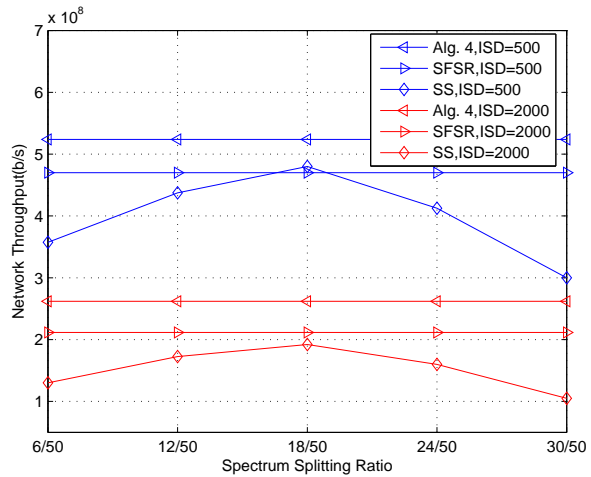}
    \caption{The network throughput of Alg. 4, SFSR, and the spectrum splitting policy.}\label{paretoten}
\end{figure}

\subsection{Convergence of Algorithm 4}
Fig. 6 shows the convergence behavior for Alg. 4, and the number of UEs in each cell changes from 10 to 50. In each iteration, the Hungarian algorithm and the DC approximation are implemented once, respectively. We set the precision to be 0.01 and simulate 1000 times. As expected, the average network throughput increases after each iteration, and Alg. 4 converges within 6 iterations.

\begin{figure}[htp]
    \centering
    \includegraphics[width=3.5in]{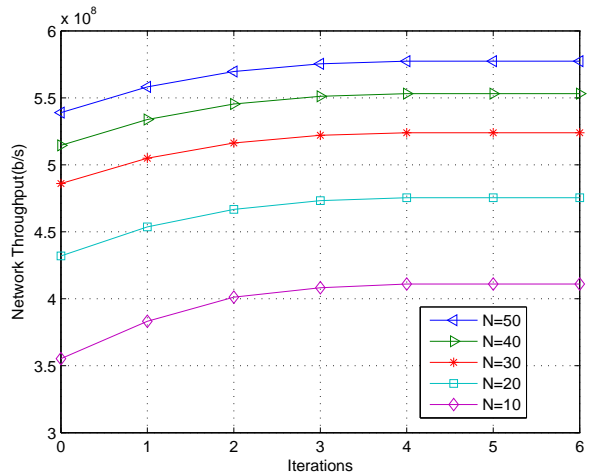}
    \caption{The average network throughput after each iteration by Alg. 4.}\label{paretotwelve}
\end{figure}

Tab. 1 shows the number of iterations to solve the convex problem in step 3 of Alg. 4, when the number of UEs changes from 30 to 150. Here, one iteration  means the power update on all the BSs. The precision is set to be 0.01. ${K_\lambda }$  and ${K_P}$ are average numbers of required iterations for updating $\bm{\lambda}$ and $\mathbf p$ respectively. ${K_T}$ is the average number of required iterations for Alg.~3. From Tab. 1, we can find that the iteration number of the Lagrange Dual method (given by ${K_\lambda }K_P$) is nearly 9 times larger than that of our method (given by ${K_T}$). Another observation is that the iteration number is almost independent of the number of UEs.
\begin{table}[htbp]\tiny
\centering
\caption{The number of iterations to solve the convex problem in step 3 of Alg. 4}
\begin{tabular}{|c|c|c|c|c|c|}
\hline
\backslashbox{parameters}{number of UEs}&30&60&90&120&150\\
\hline
${K_\lambda }$&11.314&12.015&11.941&11.736&10.928\\
\hline
$K_P$&3.829&3.906&3.514&3.716&4.037\\
\hline
$K_T$&4.235&4.419&4.137&3.819&4.386\\
\hline
\end{tabular}
\end{table}\label{tal:1}

\subsection{Multi-antenna Receivers}
In Fig. 7, we compare the network throughput of the following five methods: 1) IRC-Alg. 4: the subchannel and power allocation is achieved by Alg.4 and the UEs apply IRC receivers; 2) MRC-Alg. 4: the subchannel and power allocation is achieved by Alg.4 and the UEs apply MRC receivers; 3)IRC-SFSR: the subchannel and power allocation is achieved by SFSR and the UEs apply IRC receivers; 4) MRC-SFSR: the subchannel and power allocation is achieved by SFSR and the UEs apply MRC receivers; 5) single antenna(SA). Without loss of generality, we set the number of antennas to be 2. We observe that IRC receivers always performs better than MRC receivers under the same resource allocation method. However, if the resource allocation is implemented by Alg. 4, the gap between the two kinds of receivers will diminish due to the fact that Alg. 4 can decrease the interference from other BSs.

\begin{figure}[htp]
    \centering
    \includegraphics[width=3.5in]{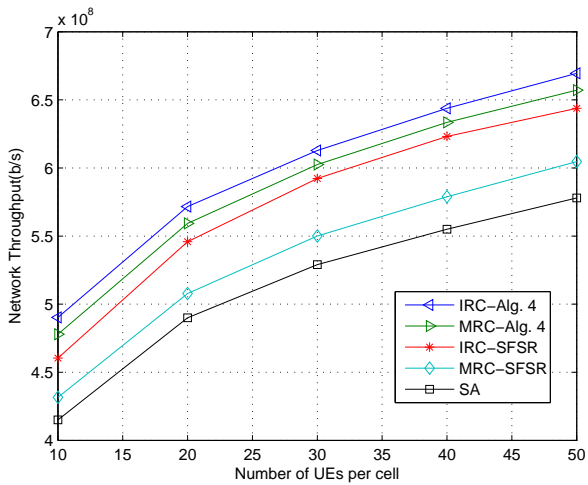}
    \caption{The network throughput versus the number of UEs per cell for the following five methods: IRC-Alg. 4, MRC-Alg. 4, IRC-SFSR, MRC-SFSR, and SA.}\label{paretotwelve}
\end{figure}

\subsection{Proportional Fairness}
The proportional fairness has been widely discussed since it strikes a good balance between network throughput and fairness by exploiting multiuser diversity and game-theoretic equilibrium \cite{fair2de2}. In this subsection, we will show the performance of our algorithm when proportional fairness is considered. It has been proven that proportional fairness can be achieved by setting the weight of each user to the reciprocal of the average rate in each time slot, when the total time slots is large enough \cite{fair2de3}. Denote the weight in time slot $i$ as

\begin{equation}\label{wee}
{\omega _k}(i) = \left\{ {\begin{array}{*{20}{c}}
   {1, \qquad \qquad \quad {\rm{}}i = 1,}  \\
   {\frac{1}{{{C_k}(i - 1)}}, \quad {\rm{  }}otherwise,}  \\
\end{array}} \right.
\end{equation}
where ${C_k}(i - 1)$ is the average rate of the $k$th UE from the first time slot to the $i - 1$th time slot. Then we can achieve proportional fairness by adopting the weight in \eqref{wee} in time slot $i$.

We focus on two scenarios: the dynamic scenario and the static scenario. In the dynamic scenario, we assume the location of the UEs are generated randomly and independently in each time slot for simplicity. While in the static scenario, the location of UEs remain static in all time slots. For each scenario, we compare the network throughput and the variance of average rate between the following four methods: 1) Alg. 4: resource allocation by Alg.4 without considering the fairness; 2) Alg. 4-PF: resource allocation by Alg.4 considering proportional fairness; 3) SFSR: resource allocation by SFSR without considering the fairness; 4) SFSR-PF: resource allocation by SFSR considering proportional fairness.

\begin{figure*}[htp]
\centering
\subfigure[Network throughput] {\includegraphics[width=3in]{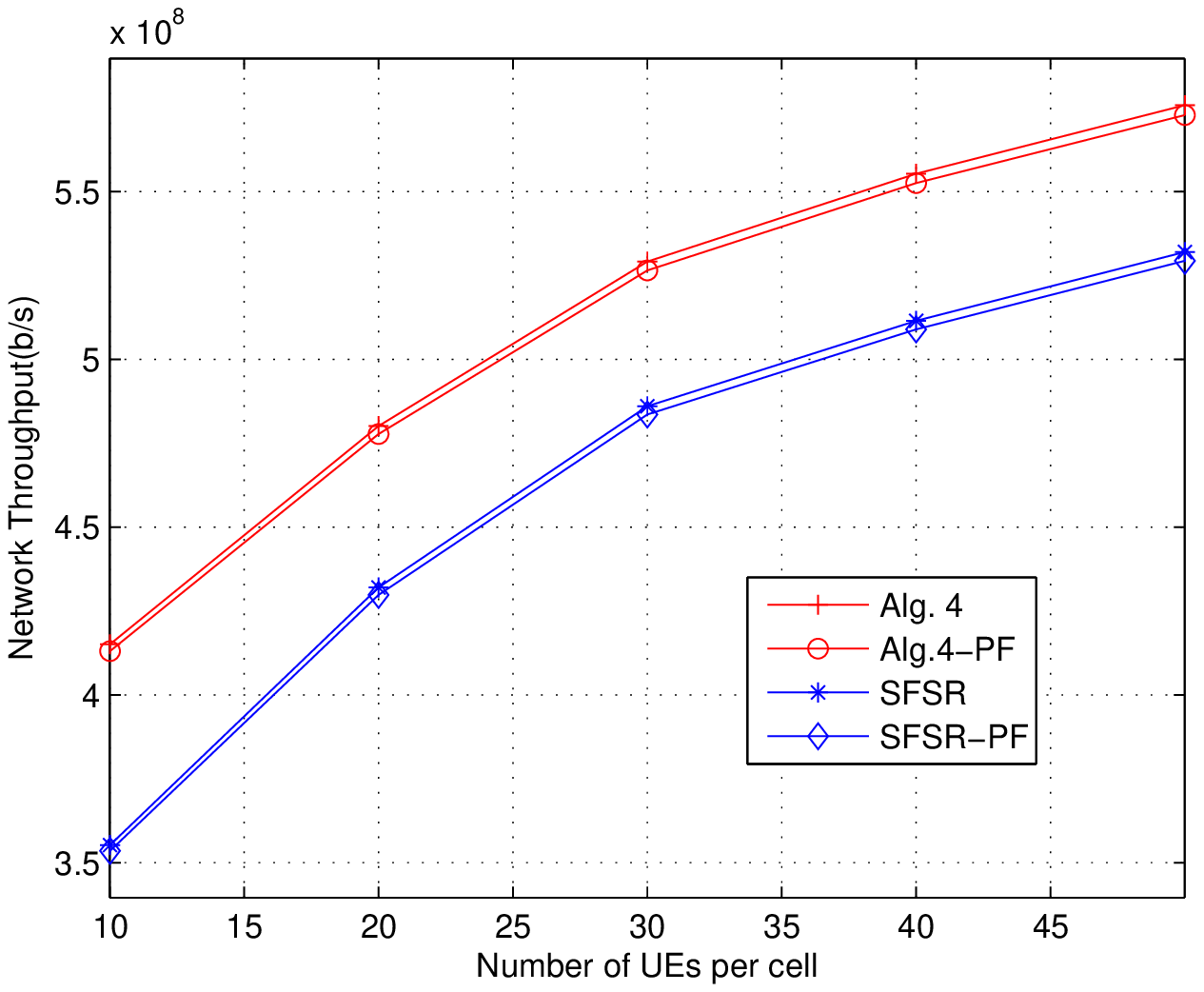}}
\subfigure[The variance of average rate] {\includegraphics[width=3in]{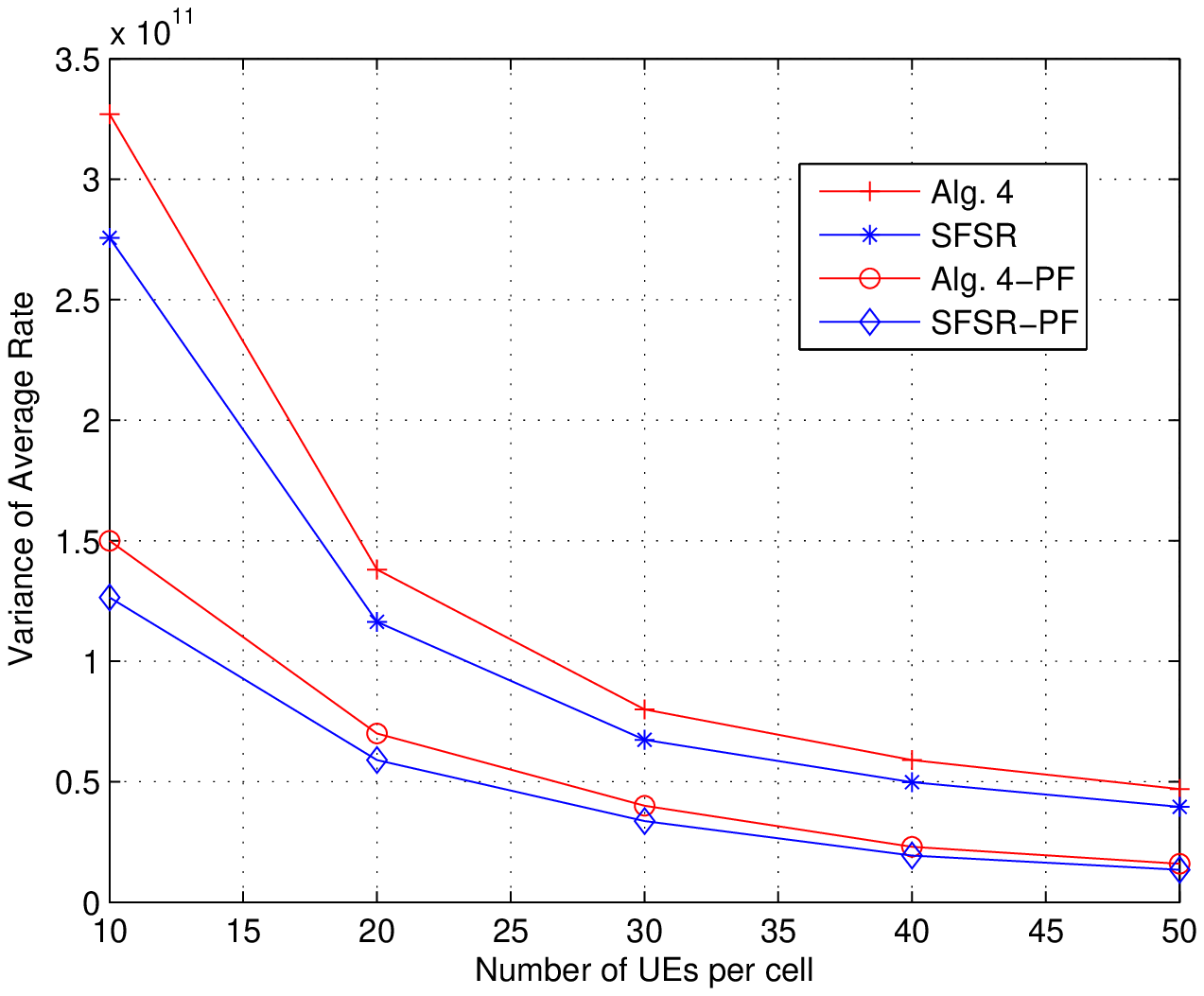}}
\caption{The comparison of network throughput and the variance of average rate in dynamic scenario.}
\label{fig5}
\end{figure*}

Fig. 8 shows the performance in the dynamic scenario, and the total time slot is  set as $1000$. From Fig. 8(a), we find that the network throughput without considering the fairness is similar to that when considering proportional fairness. The reason is that the channel coefficients of all UEs are independent identically distributed in each time slot. When the number of the time slot tends to infinity, the network throughput when considering proportional fairness will tend to that without considering the fairness. From Fig. 8(b), we observe that the variance of average rate when considering proportional fairness is nearly half of that without considering the fairness, which demonstrate the promotion of fairness. Another observation is that Alg. 4-PF performs better than SFSR in both network throughput and fairness.

Fig. 9 shows the performance in the static scenario and the total time slot is  set as $1000$. The results are obtained by averaging over 50 independent large-scale channel realizations. From Fig. 9(a), we know Alg. 4 yields a better network throughput than Alg. 4-PF. The throughput gain comes at the cost of fairness among the UEs, as shown in Fig. 9(b). Due to the fixed large-scale fading, the channel conditions  in most time  tend to be better for some UEs, while worse for others. Hence more resource is required to sustain the rate for UEs with bad channels, resulting in the decrease of the network throughput.
\begin{figure*}[htp]
\centering
\subfigure[Network throughput] {\includegraphics[width=3in]{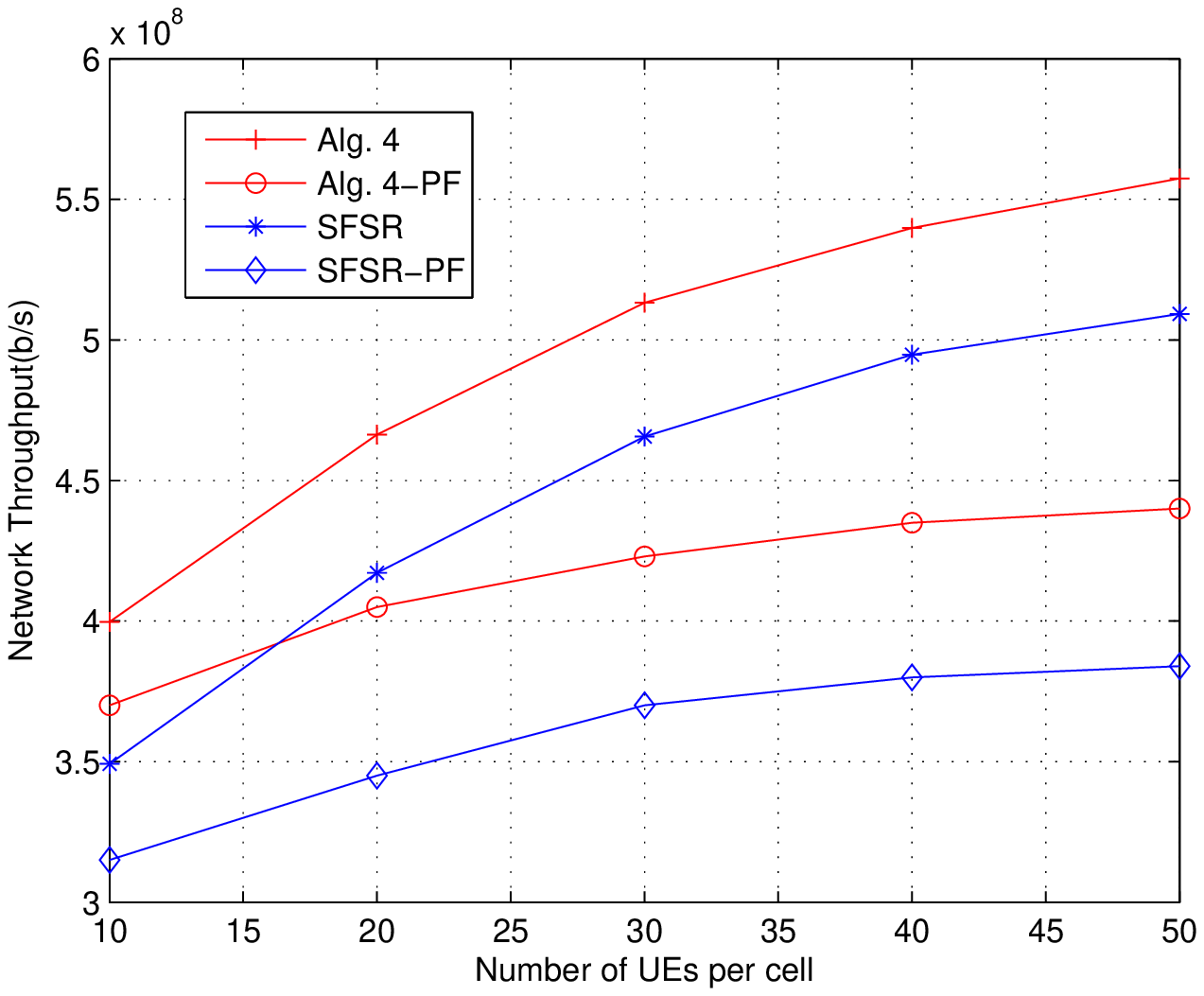}}
\subfigure[The variance of average rate] {\includegraphics[width=3in]{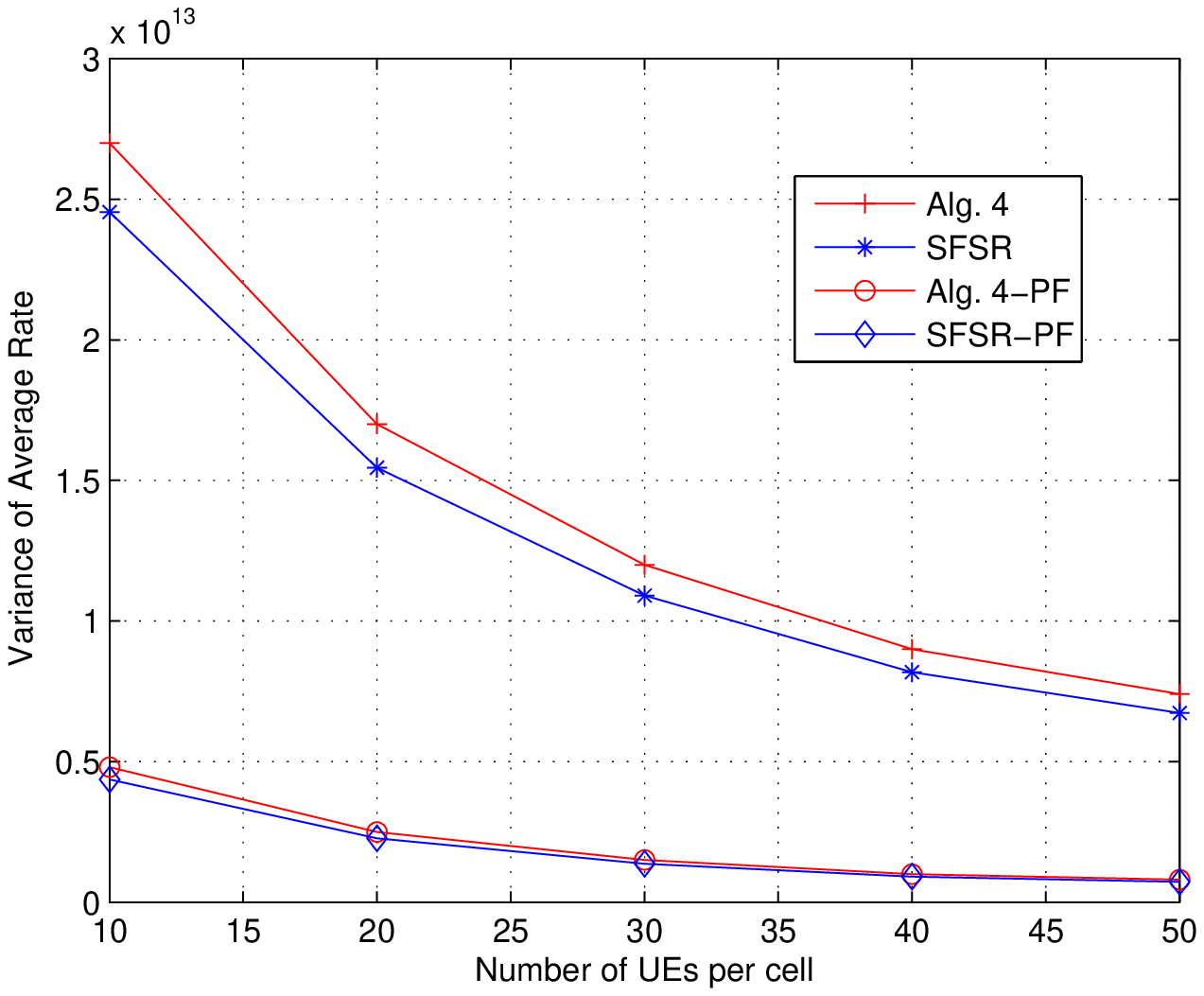}}
\caption{The comparison of network throughput and the variance of average rate in static scenario.}
\label{fig5}
\end{figure*}
So the throughput gap between Alg. 4 and Alg. 4-PF in Fig. 9(a) is obviously larger than that in Fig. 8(a). Correspondingly, the variance gap in Fig. 9(b) is obviously larger than that in Fig. 8(b). Moreover, we notice that the network throughput of Alg. 4-PF outperforms SFSR when the number of UEs is small,  while SFSR has a better performance when the number of UEs is large. This is because we only serve the UEs with superior channels for SFSR. However, for Alg. 4-PF, the number of UEs with weak channels increases with the total number of UEs, which costs more physical resource and thus inhibits the increase of the network throughput.

\section{Conclusion}
In this paper, we maximized the weighted sum-rate for the downlink transmission in Multi-cell Multi-association OFDMA heterogeneous networks. A joint user association, subchannel allocation, and power allocation optimization problem was formulated. To solve the optimization problem, we  divided it into two subproblems. The first subproblem is  joint user association and subchannel allocation for a fixed power allocation, whose globally optimal solution can be obtained by the Hungarian algorithm. The second subproblem is power allocation for  fixed  user association and subchannel allocation, which can be transformed to a series of convex problems by the DCA method. To further reduce its complexity, we proposed a simplified but efficient algorithm to solve these  problems, decreasing the number of iterations up to almost $90$ percent  off than the conventional Lagrange Dual method. Simulation results showed that our  joint optimization algorithm achieves a better performance compared with the existing algorithms. We also extend the problem into the multi-antenna receiver case and the proportional fairness case, respectively. Our algorithm  performs well when the UEs apply multi-antenna receivers, and it helps achieve a good tradeoff between throughput and fairness when considering proportional fairness.

\end{document}